\documentclass[letterpaper,11pt]{amsart}
\title{Markov Models for Accumulating Mutations
}

\author{Niko Beerenwinkel}
\address{ Department of Biosystems Science and Engineering,
ETH Zurich, Mattenstrasse 26, CH-4058 Basel, Switzerland}
\email{niko.beerenwinkel@bsse.ethz.ch}
\author{Seth Sullivant}
\address{
Department of Mathematics and Society of Fellows, 
Harvard University, Cambridge, MA 02138}  
\email{seths@math.harvard.edu}

\usepackage{latexsym,array,delarray,amsthm,amssymb,epsfig}

\theoremstyle{plain}
\newtheorem{thm}{Theorem}[section]

\newtheorem{prop}[thm]{Proposition}

\theoremstyle{definition}

\newtheorem{ex}[thm]{Example}

\theoremstyle{remark}
\newtheorem*{rmk}{Remark}

\newcommand{\rr}{\mathbb{R}}

\newcommand{\bbe}{\mathbb{E}}

\newcommand{\ind}{\mbox{$\perp \kern-5.5pt \perp$}}

\newcommand{\Prob}{\mathrm{Prob}}

\newcommand{\rmpa}{\mathrm{pa}}
\newcommand{\Exp}{\mathrm{Exp}}

\newcommand{\exit}{\mathrm{Exit}}

\usepackage[round]{natbib}

\input xy
\xyoption{all}
\CompileMatrices

\begin{document}

\begin{abstract}
We introduce and analyze a waiting time model for the accumulation of
genetic changes.  The continuous time conjunctive Bayesian network
is defined by a partially ordered set of mutations and by the rate of 
fixation of each mutation.  The partial order encodes constraints on 
the order in which mutations can fixate in the population, shedding
light on the mutational pathways underlying the evolutionary process.
We study a censored version of the model and derive equations for an 
EM algorithm to perform maximum likelihood estimation of the model 
parameters.  We also show how to select the maximum likelihood poset.  
The model is applied to genetic data from different cancers and from 
drug resistant HIV samples, indicating implications for diagnosis and 
treatment. 
\end{abstract}

\maketitle


\section{Introduction}

The genetic progression of cancer is characterized by the accumulation
of mutations in oncogenes and in tumor suppressor genes.
Recent studies have shown that during the somatic evolution of cancer
mutations in over 100 human genes are selected for, suggesting their beneficial
effect on the growth of the cancer cell \citep{Sjoeblom2006}.

In HIV infection, the virus acquires mutations in CTL epitopes that
interfere with the immune response. This evolutionary process is specific for
the genetic makeup of the infected host. Recently, a total of 478 CTL escape 
mutations have been identified in the HIV genome \citep{Brumme2007}.

Under drug treatment, HIV develops mutations that confer resistance to the
applied drugs. Eventually, this evolutionary escape leads to therapy failure.  
More than 50 drug resistance-associated mutations are known
in three different HIV proteins \citep{Johnson2006}.

These three evolutionary scenarios have in common that, for the population
of individuals, several mutations are available which increase fitness.
Adaption is therefore characterized by the accumulation of these beneficial 
mutations which are virtually non-reversible. 

In this paper, we introduce a statistical model for the accumulation of
genetic changes. The continuous time conjunctive Bayesian network (CT-CBN)
is a continuous time Markov chain model, defined by a partially ordered
set (poset) of advantageous mutations, and the rate of fixation for each mutation. 
The partial order encodes constraints on the succession in which mutations 
can occur and fixate in the population. We assume that the fixation times
follow independent exponential distributions. The exponential waiting
process for a mutation starts only when all predecessor mutations of
that mutation in the poset have already occurred. 
The order constraints and waiting times reveal important information 
on the underlying biological process with implications
for diagnosis and treatment. 

The CT-CBN is a continuous time analogue of the discrete 
conjunctive Bayesian network (D-CBN) introduced by \citet{Beerenwinkel2006a}.
The D-CBN was shown to have very desirable statistical and
algebraic properties \citep{Beerenwinkel2006e}. 
We argue that the continuous time CBN is the more natural model 
for the waiting process described above,
and we explore the connection to the discrete CBN. 
A special case of the D-CBN, where the poset is a tree, is known as
the oncogenetic or mutagenetic tree model    
\citep{Desper1999,Beerenwinkel2005f,Beerenwinkel2005b}.
It has been applied to the somatic evolution of cancer 
\citep{Radmacher2001,Rahnenfuehrer2005}
and to the evolution of drug resistance in HIV
\citep{Beerenwinkel2005a}.
The basic mutagenetic tree model has been extended to
a mixture model \citep{Beerenwinkel2005f} and to account for 
longitudinal data \citep{Beerenwinkel2007a}.

A related tree model by \citet{vonHeydebreck2004} represents the
genetic changes at the leaves of the tree and regards the interior
vertices as hidden events. Several authors have considered larger
model classes, including general Bayesian networks
\citep{Simon2000,Deforche2006} and general Markov chain models
on the state space of mutational patterns \citep{Foulkes2003b,Hjelm2006}.
As compared to trees and posets, these models are more flexible 
in describing mutational pathways, but parameter estimation and
model selection is considerably more difficult. In fact, the number of free
parameters of these models is typically exponential in the number
of mutations. By contrast, in the CT-CBN model, the number of free
parameters equals the number of mutations. We demonstrate that
parameter estimation and selection of an optimal poset can be
performed efficiently for CT-CBNs. Thus, they provide an
attractive framework for modeling the accumulation of mutations,
especially if the number of mutations is moderate or large.

We formally define the CT-CBN in the next Section~2 and derive some
basic properties of the model.  The CT-CBN is an example of a regular exponential family with closed form maximum likelihood estimates (MLEs).   In Section~3, we make precise the
relation between the CT-CBN and the D-CBN.  Section~4 deals with
a censored version of the CT-CBN which is most relevant for 
observed data.  The censored model lacks  a closed form expression for the MLE, but has a natural EM algorithm for approximating maximum likelihood estimates.   We apply our methods in Section~5 to genetic data from cancer cells and from drug resistant HI viruses. 
We close with a discussion in Section~6.

\section{Continuous Time Conjunctive Bayesian Networks}

In this section, we introduce and describe some of the basic properties of
continuous time conjunctive Bayesian networks (CT-CBN). These models are
continuous time Markov chain models on the distributive lattice of a poset.
To begin, we review some background material from combinatorics.  
The relevant combinatorial material can be found in introductory sections 
of \citet{Beerenwinkel2006a} and more detailed information is covered in 
\citet{Stanley1999}.  

A partially ordered set (poset) is a set $P$ with a 
transitive relation $\preceq$.  In our models, the set $P$ will be a set 
of genetic events, and the order relation $\preceq$ specifies partial 
information about the order in which these events must occur. The relation
$p \prec q$  implies that event $p$ happens before event $q$.  
The distributive lattice of order ideals of $P$, denoted by $J(P)$, consists 
of all subsets $S \subseteq P$, that are closed downward, i.e.,
$S \in J(P)$ if and only if, for all $q \in S$ and $p \prec q$, we have that 
$p \in S$. The order ideals of $P$ correspond to the genotypes (or
mutational patterns) that are compatible with the order constraints.
We refer to $\emptyset \in J(P)$ as the wild type.

\begin{ex}  \label{ex:running}
As a running example, consider the poset $P$ on the four element set 
$[4] = \{1,2,3,4\}$ subject to the order relations 
$1 \prec 3$, $2 \prec 3$, $2 \prec 4$.  
The distributive lattice $J(P)$ of order ideals of $P$ consists of eight
elements.  Both sets are displayed in Figure~\ref{fig:running}.  
\end{ex} 

\begin{figure}[t!]
  \centering
  \begin{tabular}{ccc}
    \begin{minipage}{3cm}
      \xymatrix@!=0.2cm{
        1\ar@{->}[d] & 2\ar@{->}[dl]\ar@{->}[d] \\
        3 & 4
      }
    \end{minipage}
    &
    \begin{minipage}{6cm}
      \begin{eqnarray*}
        & & X_1 \sim \Exp(\lambda_1) \\
        & & X_2 \sim \Exp(\lambda_2) \\
        & & X_3 \sim \max(X_1, X_2) + \Exp(\lambda_3) \\
        & & X_4 \sim X_2 + \Exp(\lambda_4)
      \end{eqnarray*}
    \end{minipage}
    &
    \begin{minipage}{5cm}
      \xymatrix@!=0.1cm{
        & \{1,2,3,4\}\ar@{-}[dl]\ar@{-}[dr] & & \\
        \{1,2,3\}\ar@{-}[dr] & & \{1,2,4\}\ar@{-}[dl]\ar@{-}[dr] & \\
        & \{1,2\}\ar@{-}[dl]\ar@{-}[dr] & & \{2,4\}\ar@{-}[dl] \\
        \{1\}\ar@{-}[dr] & & \{2\}\ar@{-}[dl] & \\
        & \emptyset & & 
      }
    \end{minipage}
    \\
    Poset, $P$ &  Waiting times, $X$ & Lattice of order ideals, $J(P)$
  \end{tabular}
  \caption{Running example.
  }
  \label{fig:running}
\end{figure}
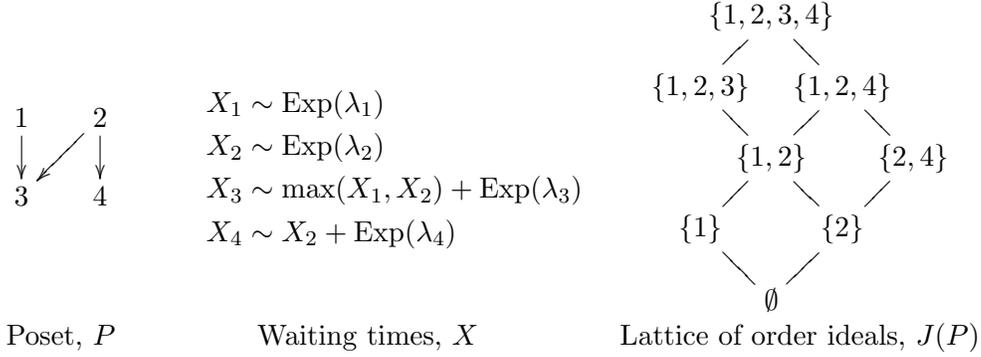
  
Let $P$ be a poset with ground set $[n]$.  For each
event $i \in P$, we define a random variable $Z_i \sim \Exp(\lambda_i)$.
Then we define the random variables $X_i$ as
\[
   X_i  =  \max_{j \in \rmpa(i)} \{ X_j \}  + Z_i,
   \quad i = 1, \dots, n.
\]
Here ${\rmpa(i)}$ is the set of all predecessors of event $i$ in the poset
$P$.  The random variable $X_i$ describes how long we have to wait until
event $i$ occurs assuming that we start at time zero with no events.  
Mutation $i$ cannot occur until all the mutations preceding it in the partial 
order $P$ have occured.  The family of joint distributions of $X$ defined in
this manner is the continuous time conjunctive Bayesian network (CT-CBN). It
has state space $\rr^n_{>0}$ consisting of vectors of waiting times and parameters 
$\lambda = (\lambda_1, \dots, \lambda_n) \in \rr^n_{> 0}$.

The probability density function associated to this model is easy to write
down, given the recursive conditional nature of the distribution.  
The density function is
\begin{equation} \label{eq:density}
  f_{P,\lambda}(t) =  \prod_{i=1}^n  \lambda_i  
    \exp (  - \lambda_i(  t_i  -  \max_{j \in \rmpa(i)} t_j) ), 
  \quad \mbox{ if } t_i > \max_{j \in \rmpa(i)} t_j 
  \mbox{ for all } i \in [n],
\end{equation}
and $f_{P,\lambda}(t) = 0$ otherwise.
The CT-CBN is an example of a regular exponential family, 
with minimal sufficient statistic consisting of the vector of time
differences  $(t_i -  \max_{j \in \rmpa(i)} t_j)_{i \in [n]}$.  

One instance of the random variable $X$ is a sequence of times 
$T = (t_1, \ldots, t_n)$ that satisfy the inequality relations implied
by $P$, i.e., $t_i > \max_{j \in \rmpa(i)} t_j$ for all $i \in [n]$.  A set of times satisfying the indicated inequality constraints is \emph{compatible} with the poset.
Thus, the data for the model is a list of sequences of times 
$T_1, \ldots, T_N$, where $N$ is the number of observations.

\begin{prop} \label{prop:ml}
Let $P$ be a poset and $T_1, \ldots, T_N$ be a collection of data.  If
any of the $T_k$ are incompatible with the poset $P$, the maximum
likelihood estimate does not exist (the likelihood function is
identically zero).  Otherwise, the maximum likelihood
estimate of $\lambda$ is given by
\[
   \widehat{\lambda}_i = 
     \frac{N}{\sum_{k =1}^N (t_{ki} - \max_{j \in \rmpa(i)} t_{kj} )  },
     \quad i \in [n].
\]
\end{prop}   

\begin{proof}
Suppose the data $T_1, \ldots, T_N$ are compatible with the poset $P$. 
Then from Equation~\ref{eq:density}, the log-likelihood function is 
\[
   \ell(\lambda_1, \dots, \lambda_n) = 
   \sum_{k =1}^N   \sum_{i =1}^n \left( \log \lambda_i 
     - \lambda_{i} (t_{ki} -  \max_{j \in \rmpa(i)} t_{kj} )  \right).
\]
Differentiating with respect to $\lambda_i$  yields the equations
$$  \sum_{k =1}^N    \left( \frac{1}{\lambda_{i}}    -  (t_{ki} -  \max_{j \in \rmpa(i)} t_{kj})  \right)  = 0,$$
and the claimed formula follows by solving for $\lambda_i$.
\end{proof}

\begin{thm}  \label{thm:MLposet}
Given data $T_1, \ldots, T_N$, the maximum likelihood poset is the
largest poset that is compatible with the data.
\end{thm}

\begin{proof}
Let $(P^1, \prec_1)$ and $(P^2, \prec_2)$ be two posets, both compatible with
the data, such that $P^1$ is a refinement of $P^2$ (that is, every relation 
that holds in $P^2$ also holds in $P^1$).   
It suffices to show that the likelihood function evaluated at the MLE is
larger for $P^1$ than $P^2$, because this implies that adding relations 
compatible with the data increases the likelihood.

Denote the MLEs for $P^1$ and $P^2$ by $\widehat{\lambda}^1$ and 
$\widehat{\lambda}^2$, respectively.  According to Proposition \ref{prop:ml} 
these values are given by
$$\widehat{\lambda}^l_i = \frac{N}{\sum_{k =1}^N (t_{ki} - \max_{j \in \rmpa_l(i)} t_{kj} ) }.$$
It does not change the expression to replace $\rmpa_l(i)$ with the set 
$\{j \in P^l \, : \, j \prec_l i \mbox{ in } P^l \}$.  
Since $P^1$ has more relations than $P^2$, this implies that the maximum is
taken over a strictly larger set, and thus 
$\widehat{\lambda}^1_i   \, \geq  \,  \widehat{\lambda}^2_i$
for all $i$.

However, the log-likelihood function evaluated at $\widehat{\lambda}^l$ is
\begin{eqnarray*}
\ell_l(\widehat{\lambda}^l \mid T)  & = & \sum_{k =1}^N   \sum_{i =1}^n \left( \log \widehat{\lambda}^l_i    - \widehat{\lambda}^l_{i} (t_{ki} -  \max_{j \in \rmpa_l(i)} t_{kj} )  \right)  \\
&   =  &  \sum_{i =1}^n \left( N \log \widehat{\lambda}^l_i  -     \widehat{\lambda}^l_i \sum_{k =1}^N 
 (t_{ki} -  \max_{j \in \rmpa_l(i)} t_{kj} )  \right) \\
 &  = &  \sum_{i =1}^n (N  \log \widehat{\lambda}^l_i  - N). 
\end{eqnarray*}
Since the logarithm is a monotone function, we deduce that
$\ell_1(\widehat{\lambda}^1 \mid T) \geq  \ell_2(\widehat{\lambda}^2 \mid T).$
\end{proof}

One of the most interesting quantities that we can compute with respect to 
the CT-CBN is the expected waiting time until a particular pattern 
$S \in J(P)$ is reached in the course of evolution.  In other words,
assuming that the parameters $\lambda$ are known, we are asking how long 
it takes until a certain collection of genetic events have occurred.
The expected waiting time is an important measure of genetic progression.
\citet{Rahnenfuehrer2005} have shown that, for mutagenetic trees, it is 
a prognostic factor of survival and time to relapse in glioblastoma and 
prostate cancer patients, respectively, even after adjustment for 
traditional clinical markers.

Note that because the exponential distributions are memoryless, 
calculating the waiting time from the wild type will also serve for 
determining the waiting time between any two patterns.  
Furthermore, the nature of the conditional factorization for the joint 
density of $X$ implies that we can restrict attention to the case 
where $S = P$, i.e., to determining the waiting time until all events 
have occurred.

Let $S \in J(P)$ be an observable genotype.  We define 
${\rm Exit}(S) = \{ j \in P \mid j \notin S,  S \cup \{j \} \in J(P) \}$
to be the set of events that have not occurred in $S$, but could
occur next.  For any subset $T \subseteq P$,  we set
$\lambda_T  =  \sum_{j \in T}  \lambda_j$.
A chain in the distributive lattice $J(P)$ is a collection of subsets 
$C_0, C_1, \ldots, C_k \in J(P)$ that satisfy $C_i \subset C_{i+1}$ 
for all $i$, with all containments strict.  A chain is maximal if it is 
as long as possible.  Note that all maximal chains in the distributive 
lattice $J(P)$ have length $n + 1$ with $n = |P|$ and start with the 
empty set as $C_0$ and reach the maximum at $C_n  = P$.  
Let $\mathcal{C}(J(P))$ denote the collection of maximal chains in $J(P)$;
a typical element is denoted $C = (C_0, \ldots, C_n)$.

\begin{thm}\label{thm:basicexpect}
The expected waiting time until all events have occurred is given by the expression
$$\bbe [  \max_{i \in P}  X_i  ] \, \,  =  \, \,    \lambda_1 \cdots \lambda_n  \sum_{C \in \mathcal{C}(J(P))}     \left( \prod_{i = 0}^{n-1}  \frac{1}{\lambda_{{\rm Exit}(C_i)} }\right)  \cdot  \left(  \sum_{i = 0}^{n-1}  \frac{1}{\lambda_{{\rm Exit}(C_i)} }\right).$$
\end{thm}

Before proving Theorem \ref{thm:basicexpect}, we want to briefly mention the idea of the proof, because it will be a common technique for proofs throughout the paper.  First of all, the indicated expectation involves the integral of a function that depends on maxima, which are not so simple to integrate directly.  So the first step is to decompose the integral into a sum of integrals over many different regions (one for each maximal chain in $J(P)$).  Over these simpler regions, the maximum function disappears.  Furthermore, these regions are each simplicial cones and the integral can then be computed by a simple change of coordinates.

\begin{proof}
Let $f(t)$ be the density function from Equation \ref{eq:density}.  We must compute
\begin{equation} \label{eq:bigexp}
  \int_{\rr^n_{\geq 0}}    \max_{i \in P}  t_i    \cdot  f(t) \, dt
  \end{equation}
Let $S_n$ denote the symmetric group on $n$ letters with $\sigma = (\sigma_1,
\ldots, \sigma_n) $ a typical element.  The integral (\ref{eq:bigexp}) over
the positive orthant breaks up as the sum
$$
\sum_{\sigma \in S_n}  \int_{t_{\sigma_1} = 0}^\infty  \int_{t_{\sigma_2} = t_{\sigma_1}}^\infty \cdots   \int_{t_{\sigma_n} = t_{\sigma_{n-1}}}^\infty   t_{\sigma_n}  f(t)  dt.
$$
That is, the sum breaks up the integral into smaller integrals over regions 
$$0 < t_{\sigma_1} < t_{\sigma_2} < \cdots < t_{\sigma_n}.$$
The integrand is zero unless $\sigma_1, \sigma_2, \ldots, \sigma_n$ is a linear extension of the poset $P$.  In other words, the integrand is zero unless the sets
$C_i   =  \cup_{j = 1}^i  \{ \sigma_j \}$ for $i = 0, \ldots, n$
form a maximal chain in the distributive lattice $J(P)$.  So suppose that $\sigma$ is a linear extension of $P$.  Without loss of generality, we may suppose that this linear extension is $1  \prec  2 \prec \cdots \prec n$.

We must compute the integral
$$
\int_{t_{1} = 0}^\infty  \int_{t_{2} = t_{1}}^\infty \cdots   \int_{t_{n} = t_{{n-1}}}^\infty   t_{n}  f(t) \, dt
$$
where over this restricted region, $f(t)$ now has the form
$$ f(t)  =   \prod_{i =1}^n   \lambda_i    \exp (  - \lambda_{i}(t_{i}  - t_{j(i)})) $$
where $j(i)$ is the largest number with $j(i) \prec i$ in $P$.  

Now introduce the change of coordinates
$$u_0 =  t_{1}, \,\,   u_i  =  t_{{i+1} }  -  t_{i} \mbox{ for } i = 1, 2, \ldots, n-1.$$
The determinant of this linear transformation is one, so the integral becomes
$$
\int_{u_0 = 0}^\infty \int_{u_1 = 0}^\infty  \cdots  \int_{u_{n-1} = 0}^\infty  (u_0 + \cdots +  u_{n-1} )  \prod_{i =1}^n  \lambda_i  \exp(-\lambda_{i} (u_{i-1} +  u_{i-2}  + \cdots + u_{j(i)}  ) )  \, du. 
$$
The multiple integral is now over a product domain, and involves a function 
in product form, so we want to break this integral up into the product of integrals.  
To do this, we must collect the $\lambda_{i}$ terms that go with the various $u_k$ terms.  
In the exponent, we have that $\lambda_i$ appears as a coefficient of $u_k$ 
if and only if $i > k \geq j(i)$.  This, in turn, implies that when all the events 
$1,2, \ldots, k$ have occurred, all the predecessor events of $i$ have occurred.  
But this means that $i \in {\rm Exit}(C_k)$, where $C_k = \{1,2\ldots, k \}$.   
Thus, the transformed integral breaks up as a sum of $n$ integrals that have the form:
$$\lambda_1 \cdots \lambda_n  \cdot \int_{u_0 = 0}^\infty \int_{u_1 = 0}^\infty  \cdots  \int_{u_{n-1} = 0}^\infty  u_j   \prod_{i = 0}^{n-1}  \exp(  - \lambda_{{\rm Exit}(C_i)} u_i)  \, du.$$
The integral is over a product domain of a product function.  By elementary integration, it is
$$\lambda_1 \cdots \lambda_n  \frac{1}{\lambda_{{\rm Exit}(C_j)}}  \prod_{i = 0}^{n-1}  \frac{1}{\lambda_{{\rm Exit}(C_i)}},$$
which  completes the proof.
\end{proof}


\section{Relation to the Discrete Conjunctive Bayesian Network}

In this section, we explore the connection between the CT-CBN and the discrete CBN (D-CBN) introduced in \citet{Beerenwinkel2006e}.  
Part of the motivation for this project was to understand how the two types of 
models relate to each other and how structural information from one model 
can be used to deduce information about the other.  Also, we are naturally led 
to study discrete models because we rarely have access to the times at which 
the individual events occurred, but can only check, after a certain sampling
time, which of the events have occurred.

We will show that the D-CBN gives a first order approximation to the transition probabilities in the CT-CBN.  This suggests that
the D-CBN is not optimal from a modeling standpoint as the nature of our
applications is to wait until mutations occur.  On the other hand, the D-CBN 
is much simpler to work with, and maximum likelihood estimates for the D-CBN 
can be used as a first step for iterative algorithms for ML estimation in the 
censored versions of the CT-CBN, described in Section~4.

To explain the first order approximation result, we consider the CT-CBN as a 
continuous time Markov chain on the distributive lattice $J(P)$ (see 
\citet{Norris1997} for background on Markov chains).  
The rate matrix for the Markov chain is the upper-triangular $m \times m$
matrix $Q$, where $m = |J(P)|$. The entries of $Q$ are indexed by pairs of sets
of occurred events $S, T \in J(P)$ and are given by
\[
   Q_{S,T}   =  \left\{  \begin{array}{ll}
       \lambda_{j}  &   \mbox{ if }   S \subset T \mbox{ and } T \setminus S = \{j \}, \\
       - \lambda_{{\rm Exit}(S)}  &  \mbox{ if }  S = T,  \mbox{ and } \\
       0 &  \mbox{ otherwise}.  
     \end{array}  \right.
\]
If we fix a linear extension of $J(P)$ and order the rows and columns of $Q$ according to this linear extension, $Q$ will be an upper triangular matrix.

Let $p(t)$ be the $m \times m$ matrix where, for $S, T \in J(P)$,  
the entry $p_{S,T}(t)$ denotes the probability that the continuous time Markov
chain with state space $J(P)$, is in state $T$ at time $t$ starting from state 
$S$ at time $0$.  This quantity can be calculated by integrating the density 
function from the continuous time model.
However, it is simpler to calculate from standard theory of Markov chains.  
Indeed, the matrix $p(t)$ is the solution to the system of differential equations
$$
\frac{d}{dt}  p(t) =  Q p(t)   
$$
subject to the initial conditions $p(0)  =  I$, the identity matrix.  
The solution to these first order differential equations is obtained by taking 
$p(t)  =  \exp(Qt)$, where $\exp$ denotes the matrix exponential:
$$\exp(Qt)  =  I  + Qt +  \frac{Q^2t^2}{2!}  +   \frac{Q^3t^3}{3!}  +  \cdots .$$ 
Clearly, the solution to the differential equations then satisfies
$$ \frac{d^k}{dt^k} p(t)  |_{t = 0}   =  Q^k$$
so a first order approximation to $p(t)$ is a function $\tilde{p}(t)$ that satisfies
$$\tilde{p}(0)  =  I  \quad \mbox{and}  \quad  \frac{d}{dt} \tilde{p}(t)  |_{t =0}   = Q.$$  

A first order approximation to the CT-CBN can be derived from the D-CBN as
follows.  Associated to the D-CBN are $n$ parameters 
$\theta_1, \ldots, \theta_n$,  where $\theta_i$ is the conditional probability 
that event $i$ has occurred, given that all its predecessor events have
occurred.   By setting
$\theta_i  =  1 - \exp( - \lambda_i t)$ and defining
\[
   \widetilde{p}_{S,T}(t)  = 
   \prod_{i \in T \setminus S} \theta_i \prod_{i \in {\rm Exit}(T)}  (1  - \theta_i),
   \quad \mbox{ if }  S \subseteq T,
\]
and $\widetilde{p}_{S,T}(t) = 0$ otherwise, the D-CBN is naturally interpreted as a continuous time model, where $\widetilde{p}_{S,T}(t)$ is the probability that the D-CBN is in state $T$ at time $t$ given a starting state of $S$ at time $0$.

\begin{prop}
The model $\tilde{p}(t)$ derived from the D-CBN is a first order approximation to the 
CT-CBN $p(t)$.  
\end{prop}

\begin{proof}
Note that diagonal entries of $\tilde{p}(t)$ have the form $\prod \exp( -\lambda_i t)$, and off-diagonal entries are either identically zero or are a product of terms, at least one of which is of the form $1 - \exp( - \lambda_i t) $.  This implies that $\tilde{p}(0) = I$.

If $T = S$, then $\widetilde{p}_{S,S}(t)  =  \prod_{i \in {\rm Exit}(S)}  \exp(-\lambda_i t)$.  So, by the product rule  
$$\frac{d}{dt} \widetilde{p}_{S,S}(t)  =   \sum_{i \in {\rm Exit}(S)} - \lambda_i  \widetilde{p}_{S,S}(t)$$
and thus    $\frac{d}{dt} \widetilde{p}_{S,S}(t) |_{t = 0}  = - \lambda_{{\rm Exit}(S)}  =  Q_{S,S}$.  If $S \subset T$ then 
$$\tilde{p}_{S,T}(t)  =  \prod_{i \in T \setminus S}  (1 - \exp(-\lambda_i t) ) \cdot  \prod_{i \in {\rm Exit}(T)}  \exp( - \lambda_i t) $$
and thus
$$\frac{d}{dt}  \tilde{p}_{S,T}(t)   =  \sum_{i \in T \setminus S}  \lambda_i  \frac{\exp( -\lambda_i t)}{ 1 - \exp( - \lambda_i t) }   \tilde{p}_{S,T}(t)  -  \sum_{i \in {\rm Exit}(T)} \lambda_i  \tilde{p}_{S,T}(t).$$

If $T \setminus S$ has cardinality greater than one, then $\frac{d}{dt}
\tilde{p}_{S,T}(t) |_{t = 0} =0  = Q_{S,T}$, since every term in the sum
involves at least one expression of the form $1 - \exp( - \lambda_i t)$.  On
the other hand, if $T = S \cup \{j\}$, then $\frac{d}{dt} \tilde{p}_{S,T}(t)
|_{t = 0} = \lambda_j  = Q_{S,T}$, since only the first term in the sum does
not contain an expression of the form $1 - \exp( -\lambda_i t)$.  As all other
entries in $\tilde{p}(t)$ and $Q$ are zero, this proves that $\widetilde{p}(t)$
is a first order approximation to $p(t)$.
\end{proof}

Given that the D-CBN is a first order approximation to the CT-CBN, 
it seems natural to conjecture that these two models are, in fact, equal to
each other. Indeed, if $P$ is the poset with no relations, 
it is easy to show that these two models coincide.  However, if $P$ contains 
at least one relation, the models are no longer the same, 
and the D-CBN is not even a second order approximation of the CT-CBN. 
This is illustrated in the following example.

\begin{ex}
Let $P$ be the poset on two elements with one relation $1 \prec 2$ and
fix the natural order $\emptyset$, $\{1\}$, $\{1,2\}$ in $J(P)$.  
If $\lambda_1 \neq \lambda_2$, then
$$
\tilde{p}(t)  =  \begin{pmatrix}
e^{-\lambda_1 t}  &  (1 - e^{ - \lambda_1 t} ) e^{ -\lambda_2 t}  &   (1 - e^{ - \lambda_1 t} ) (1 - e^{ - \lambda_2 t} )   \\
0 &  e^{- \lambda_2 t}  &  1 - e^{ - \lambda_2 t}  \\
0 & 0 & 1  \end{pmatrix}$$
and
$$
p(t)  =  \begin{pmatrix}
e^{- \lambda_1 t}   &    \frac{\lambda_1}{ \lambda_1 - \lambda_2} (e^{- \lambda_2 t} -  e^{- \lambda_1 t})
&  1  -  \frac{  \lambda_1 e^{-\lambda_2 t}  -  \lambda_2  e^{-\lambda_1t}}{ \lambda_1 - \lambda_2}  \\ 
  0 &  e^{- \lambda_2 t}  &  1 - e^{ - \lambda_2 t}  \\
0 & 0 & 1  \end{pmatrix}.
$$
In particular, $\frac{d^2}{dt^2} \tilde{p}_{\emptyset, \{1\}}(t) |_{t = 0}  =  -\lambda_1^2 - 2 \lambda_1 \lambda_2$,  whereas  $(Q^2)_{\emptyset, \{1\}}  =  -\lambda_1^2 -  \lambda_1 \lambda_2$.  \end{ex}

The discrepancies exhibited in the previous example become more dramatic as
the poset $P$ develops longer chains.  While the D-CBN is not identical to 
the CT-CBN, its nice properties can be exploited at various points during 
optimization.


\section{Censoring}

Our goal in this section is to define and analyze a censored CT-CBN model 
that we will apply to genetic data in Section~5.  The reason for introducing 
models with censoring is that we rarely know explicitly the time points
$t_1, \dots, t_n$ at which the events have occurred. Often we can only
measure, at a particular time, which of the events have occurred so far.  
It is natural to assume that the observation times are themselves random.  
For example, the evolutionary process leading to drug resistance in HIV
starts with the onset of therapy.  However, the genome of the virus can only 
be determined after viral rebound, i.e., after loss of viral supression,
which typically involves several mutations. Thus, the sampling time is
the time to therapy failure.

We introduce a new event $s$, such that the random variable $X_s$ is an
independent, exponentially distributed stopping time 
(or sampling time, or observation time), $X_s  \sim { \rm Exp}(\lambda_s)$.  
We define a new poset $P_s = P \cup \{s\}$ by adding to the poset of mutations
the element $s$, which has no relations with the other elements in $P$.  
In this setting, the events we observe consist of subsets $S \subseteq [n]$, 
which correspond to the event that $X_i < X_s$ for all $i \in S$ and 
$X_s < X_i$ for all $i \in [n] \setminus S$.

Thus, an observed set of events $S$ imposes extra relations 
$i \prec s$ for $i \in S$ and $s \prec i$ for $i \in [n] \setminus S$
on the poset $P_s$,
and we are led to study the poset refinements of $P_s$.
A poset $Q$ is said to refine the poset $P_s$ if every relation in $P_s$ 
also holds in $Q$.  A realization of the random vector $X$ is said to be 
compatible with $Q$, denoted $X \vdash Q$, if $X_i < X_j$ whenever 
$i \prec j$ in $Q$.  We can directly compute the probability of the event 
$X \vdash Q$ in terms of the distributive lattices $J(Q)$ and $J(P_s)$. 
Throughout this section, we abuse notation and say that $X \sim P_s$ if $X = (X_s, X_1, \ldots, X_n)$ is distributed according the CT-CBN associated to poset $P_s$ with parameter vector $\lambda = (\lambda_s, \lambda_1, \ldots, \lambda_n)$.

\begin{thm}\label{thm:prob}
The probability that $X \sim P_s$ is compatible with the poset $Q$ is given by
\begin{equation}\label{eq:prob}
\Prob(X \vdash Q) = 
 \lambda_s \lambda_1 \cdots \lambda_n \sum_{C \in \mathcal{C}(J(Q))}  \prod_{i=0}^{n}  \frac{1}{
  \lambda_{\exit(C_i)} },
\end{equation}
where the sum runs over all maximal chains in the distributive lattice $J(Q)$.
\end{thm}  

\begin{rmk}
Note that in this formula, and in all formulas throughout this section, the expression ${\rm Exit}(S)$ always refers to the underlying poset $P_s$ and not to the refinement $Q$.
\end{rmk}

\begin{proof}
We must compute the integral 
$$ \int_{\rr^{n+1}_{\geq 0} }    \mathbb{I}_Q(t) \cdot f(t)  dt$$
where $\mathbb{I}_Q(t)$ is the indicator  function of compatibility with the poset $Q$.  The integral breaks up into a sum over the linear extensions of the poset $Q$ over regions over the form $t_{\sigma_0} < t_{\sigma_1} < \cdots < t_{\sigma_n}$.  Without loss of generality and after renaming the elements of the poset $P_s$, we can assume that the linear extension of interest is $0 \prec 1 \prec \cdots \prec n$.  We must calculate the integral
$$\int_{t_0 = 0}^\infty \int_{t_1 = t_0}^\infty \cdots  \int_{t_n = t_{n-1}}^\infty  f(t)  dt$$
where over the restricted region, the integrand has the form
$$ f(t)  =   \prod_{i =0}^n   \lambda_i    \exp (  - \lambda_{i}(t_{i}  - t_{j(i)})). $$
Taking the usual change of variables $u_{0}  =  t_0$  and $u_{i+1}  = t_{i+1} - t_i$ for $i = 1, \ldots, n$ we see that the integral becomes
$$  \prod_{i = 0}^n    \int_{ u_{i} = 0}^\infty   \exp(  -\lambda_{{\rm Exit}(C_i) } u_i) du_i $$
where $C_i  =  \{0,1, \ldots, i-1 \}$.    This yields the desired contribution to the integral.
\end{proof}

\begin{ex}
Let $P$ be the poset from Example~\ref{ex:running} with relations 
$1 \prec 3$, $2 \prec 3$, and $2 \prec 4$, and consider the extended
poset $P_s = P \cup \{s\}$ with no additional relations.  
Suppose we want to calculate the probability of precisely mutations $2$ and $4$ 
occurring before measurement.  The refinement $Q_{2,4}$ corresponding to the
genotype $\{2, 4\}$ is a chain $2 \prec 4 \prec s \prec 1 \prec 3$, 
and so the distributive lattice $J(Q_{2,4})$ is also a chain.  From
Equation~\ref{eq:prob} we see that  
$$\Prob(X \vdash Q_{2,4})  =     \lambda_1 \lambda_2 \lambda_3 \lambda_4  \lambda_s  \frac{1}{\lambda_1 + \lambda_2 + \lambda_s}  \frac{1}{\lambda_1 + \lambda_4 + \lambda_s}  \frac{1}{\lambda_1 + \lambda_s} \frac{1}{ \lambda_1}  \frac{1}{\lambda_3}.$$
On the other hand, the distributive lattice $J(Q_{1,2})$ has four chains and $\Prob(X \vdash Q_{1,2})$ is the sum of four terms of product form.  The terms in the sum have common factors, and this expression can be rewritten as
$$\Prob(X \vdash Q_{1,2})  =  \frac{\lambda_1 \lambda_2 \lambda_3 \lambda_4 \lambda_s}{
(\lambda_1 + \lambda_2 +  \lambda_s)( \lambda_3+ \lambda_4+ \lambda_s)( \lambda_3+ \lambda_4)}  \left(  \frac{1}{ \lambda_2+ \lambda_s} +  \frac{1}{ \lambda_1+ \lambda_4 +  \lambda_s}  \right)  \left(  \frac{1}{ \lambda_3}  + \frac{1}{ \lambda_4}  \right).$$
The method for building  the expression for this probability recursively is explained in Proposition \ref{prop:dynamic}. \qed
\end{ex} 

Given a poset $P$ and a set of parameters $\lambda_1, \ldots, \lambda_n$,
and $\lambda_s$, we obtain a probability distribution 
$p(P, \lambda)  \subseteq  \Delta_{2^n}$, the $2^n -1$ dimensional probability
simplex.  The set of all such probability distributions is the discrete
censored CT-CBN.   Although there are $n+1$ parameters specifying the model, 
it is easy to see that the family of probability distributions that can arise 
has dimension $n$.  The loss of dimensionality arises from the fact that 
$p(P, \lambda)  = p(P, t \lambda)$.  This is an example of an algebraic statistical model \citep{Pachter2005}, since the probability $p(P,\lambda)$ is a rational function of the parameters $\lambda$.  Unfortunately, these models seem difficult to analyze using techniques from algebraic geometry.  Indeed, even the model associated to the poset with no relations, corresponding to independently occurring mutations, lacks a simple description as an algebraic statistical model.

Unlike in the fully observed CT-CBN or the D-CBN, we have found no general closed form 
expressions for the MLEs of the parameters of the censored CT-CBN.  
However, as the censored model is a marginalization of the CT-CBN and 
the CT-CBN is a regular exponential family, we can use the EM algorithm 
to find ML estimates (see Chapter 8 of \citet{Little2002} for a general description of the EM algorithm).  While the EM algorithm is only 
guaranteed to find a local maximum of the likelihood function, our
computational experience has been that using exact MLEs from the D-CBN 
for $\theta_i$ and solving for $\lambda_i$ in the approximate expression 
$\theta_i = \lambda_{i} / (\lambda_i + \lambda_s)$ works as a good 
starting guess for the EM algorithm.

In the EM algorithm for a marginalization of a regular exponential family, 
we start with a guess for the ML parameters $\lambda^*$.  Then, given the
data, we compute the expected values for the missing sufficient statistics 
of the fully observed regular exponential family.  In our setting, 
we need to compute, for each $S \subseteq [n]$ that is observed, 
and for each $i \in [n]$, the expected value
\begin{equation}  \label{eq:exp}
   \bbe[X_i - \max_{j \in {\rm pa}(i)} X_j \mid X \vdash Q_S].
\end{equation}
This is the \emph{E-step} of the EM algorithm.
The expected sufficient statistics are then used to compute MLEs for 
$\lambda$ in the fully observed CT-CBN.  This is the \emph{M-step} 
of the EM algorithm.  The EM-algorithm iterates alternations of the E-step and
the M-step.  After each iteration, the likelihood function is guaranteed to
increase.  A fixed point of the EM algorithm must be a critical point of the
likelihood function.  Running the EM-algorithm with many starting points is a
useful heuristic for calculating maximum likelihood estimates.  

Since the M-step in our EM algorithm is trivial to calculate from 
Proposition~\ref{prop:ml}, the only thing remaining to compute is the
expected value (\ref{eq:exp}) .  The formula is similar to the one that appears in
Theorem~\ref{thm:basicexpect}.

\begin{thm}\label{thm:Estep}
The expected value of $X_i - \max_{j \in {\rm pa}(i)} X_j$ given that $X
$ is compatible with $Q$ is 
$$
\bbe[ X_i - \max_{j \in {\rm pa}(i)} X_j |  X \vdash Q]  =  \frac{\lambda_s\lambda_1 \cdots \lambda_n}{\Prob(X \vdash Q)}    \sum_{C \subset \mathcal{C}(J(Q))}  \left( \left(    \prod_{k = 0}^{n}  \frac{1}{ \lambda_{\exit(C_k)} } \right) \cdot  \left( \sum_{k =0}^{n}  \frac{ \iota(i, C_k)}{\lambda_{\exit(C_k)}} \right) \right)
$$
where the first sum is over all maximal chains in $J(Q)$ and
$$\iota(i,C_k) =  \left\{  \begin{array}{cl}   
1 &  \mbox{ if  } i \notin C_k \mbox{ and }  {\rm pa}(i)  \subseteq C_k \\
0 &  \mbox{ otherwise}.  \end{array}  \right. $$
\end{thm}

\begin{proof}
The proof follows the same basic outline of the proof of Theorem \ref{thm:basicexpect}.
The expected value is
$$
\bbe[ X_i - \max_{j \in {\rm pa}(i)} X_j |  X \vdash Q]  =   \frac{1}{{\rm Prob}(X  \vdash Q) }  \int_{\rr^{n+1}_{\geq 0}}     (t_i -  \max_{j \in {\rm pa}(i)} t_j) \cdot  \mathbb{I}_Q(t) \cdot f(t)  dt.
$$
We can calculate the integral by decomposing it into a sum over the linear extensions of $Q$, i.e., the chains in the distributive lattice $J(Q)$.  Without loss of generality, we can suppose that the linear extension is called $0 \prec 1 \prec \cdots \prec n$.  For this linear extension, the integral becomes
$$
\int_{t_0 = 0}^ \infty  \int_{t_1 =  t_0}^\infty  \cdots  \int_{t_n = t_{n-1}}^\infty  (t_i -  t_{j(i)}) f(t)  dt
$$
and over this region $f(t)$ has the form
$$ f(t)  =   \prod_{k =0}^n   \lambda_k    \exp (  - \lambda_{k}(t_{k}  - t_{j(k)})). $$
Applying the usual change of coordinates, we can rewrite this integral in product form as
$$
\int_{u_0 = 0}^\infty \int_{u_1 = 0}^\infty  \cdots  \int_{u_{n} = 0}^\infty  (u_{i-1}  +  \cdots + u_{j(i)} )  \prod_{k =0}^n  \lambda_k  \exp(-\lambda_{k} (u_{k-1} +  u_{k-2}  + \cdots + u_{j(k)}  ) )  \, du. 
$$
Breaking this integral up as a sum, yields a collection of integrals we have already computed in the proof of Theorem \ref{thm:basicexpect}.  However, each term 

$$ \left( \prod_{k = 0}^{n}  \frac{1}{ \lambda_{\exit(C_k)} }  \right)  \cdot \frac{1}{ \lambda_{\exit(C_k)} }$$
contributes to the sum if and only if $i \in {\rm Exit}(C_k)$, i.e., if and only if $\iota(i,C_k) = 1$.
\end{proof}

Rather than computing the expectation from Theorem~\ref{thm:Estep} by
explicitly listing all maximal chains in the distributive lattice $J(Q)$,
the expected value can be computed recursively, by summing up the
distributive lattice.  This approach is sometimes referred to as \emph{dynamic
programming}.  It reduces the computational burden of computing the
expectation, because one need not enumerate all the chains in $J(Q)$.

\begin{prop}\label{prop:dynamic}
For each $S \in J(Q)$ define $P_S$ and $E^i_S$ by the formulas
$$P_S  =  \sum_{j \in S :  S \setminus \{j \} \in J(Q)}  \frac{ \lambda_j}{ \lambda_{\exit(S \setminus \{j \})}}  P_{S \setminus \{j \}}$$
$$E^i_S  =  \sum_{j \in S :  S \setminus \{j \} \in J(Q)}   \left( \frac{ \lambda_j}{ \lambda_{\exit(S \setminus \{j \})}}  E^i_{S \setminus \{j \}}  + \iota(i, S \setminus \{j\})  \frac{ \lambda_j}{ \lambda^2_{\exit(S \setminus \{j \})}}  P_{S \setminus \{j \}}  \right),$$
subject to the initial conditions $P_\emptyset = 1$ and $E^i_\emptyset = 0$, where 
$$\iota(i, S \setminus \{j\}) =  \left\{  \begin{array}{cl}   
1 &  \mbox{ if  } i \notin S \setminus \{j \}  \mbox{ and }  {\rm pa}(i)  \subseteq S \setminus \{j \}  \\
0 &  \mbox{ otherwise}.  \end{array}  \right. $$
Then $\Prob(X  \vdash Q)  =  P_{\{s \}  \cup [n]}$ and 
$\bbe[X_i - \max_{j \in {\rm pa}(i)}  X_j \,  | \,  X \vdash Q]  =  E^i_{\{s \} \cup [n]} \big/ P_{\{s \} \cup [n]}$.
\end{prop}

\begin{proof}
 Both results follow from writing down a closed form for $P_S$ and $E^i_S$, proving that these formulas hold inductively, and showing that $P_{\{s \} \cup [n]}  =  \Prob(X \vdash Q)$ and $ \frac{E^i_{\{s \} \cup [n]}}{P_{\{s \} \cup [n]}}  = \bbe[X_i - \max_{j \in {\rm pa}(i)}  X_j \,  | \,  X \vdash Q]  $.

To this end, let $Q|_S$ be the induced subposet of $Q$ with element set $S$.  Then
$$P_S  =  \prod_{i \in S} \lambda_i 
\sum_{C \in \mathcal{C}(J(Q_S))}  \prod_{i=0}^{|S| -1}  \frac{1}{
  \lambda_{\exit(C_i)} }$$
with $P_\emptyset = 1$.  The recurrence 
$$P_S  =  \sum_{j \in S :  S \setminus \{j \} \in J(Q)}  \frac{ \lambda_j}{ \lambda_{\exit(S \setminus \{j \})}}  P_{S \setminus \{j \}}$$
is satisfied because every maximal chain in $J(Q|_S)$ comes from a maximal chain in exactly one of the $J(Q|_{S \setminus \{j \}})$ by adding $j$ to the poset $Q|_{S \setminus \{j \}}$ as the last element.  Also, $P_{\{s\} \cup [n]}$ has the desired form.

Similarly, it is straightforward to show that 
$$E^i_S  = 
\prod_{k \in S}  \lambda_k    \sum_{C \subset \mathcal{C}(J(Q_S))}  \left( \left(    \prod_{k = 0}^{|S|-1}  \frac{1}{ \lambda_{\exit(C_k)} } \right) \cdot  \left( \sum_{k =0}^{|S|-1}  \frac{ \iota(i, C_k)}{\lambda_{\exit(C_k)}} \right) \right)$$
which, together with $P_S$, above, satisfies the desired recurrence relation.
\end{proof}


\section{Applications}

\begin{table}[h!]
  \centering \footnotesize
  \begin{tabular}{llrrl}
    \hline
    \bf Biological system & \bf Genetic events & \bf \# & \bf Samples & \bf Ref. \\ \hline
    Prostate cancer & Chromosomal gains and losses & 9 & 54 & \cite{Rahnenfuehrer2005} \\
    Colon cancer & Mutated genes & 12 & 35 & \cite{Sjoeblom2006} \\
    Breast cancer & Mutated genes & 9 & 42 & \cite{Sjoeblom2006} \\
    HIV drug resistance & Amino acid changes in the HIV RT & 7 & 364 & \cite{Beerenwinkel2005a} \\ \hline
    ~ \\
  \end{tabular}
  \caption{Applications of the CT-CBN model to genetic data. For each biological system
  (first column), the nature (second column) and the number (third column) 
  of the genetic alterations, the number of observations (fourth column), and
  a reference pointing to the original study (last column) is shown.
  }
  \label{tab:applications}
\end{table}

In this Section, we use the CT-CBN model to describe the accumulation of
mutations in four different biological systems (Table~\ref{tab:applications}). 
The random variables $X_i$ denote the times of fixation of genetic changes 
in a population of individuals. The definition of a genetic change is 
different in each example, depending on both the nature of the 
evolutionary process and the technology to detect genetic alterations.
For all examples, the data is a list of genotypes that have been observed 
after an unknown sampling time assumed to be exponentially distributed. 
Thus we apply the censored CT-CBN model.
Our goal is to learn the structure of mutational pathways, which is
represented by the linear extensions of the CT-CBN defining posets.

Theorem~\ref{thm:MLposet} states that the structure of the ML CT-CBN model
is given by the maximal poset that is compatible with the observed data.
In practice, however, the observations are subject to noise, either
due to deviations of the data generating process from the model, 
or due to technical limitations in assessing genetic changes.
Thus, for most biomedical data sets, the ML poset will have very few
relations, although a large portion of the observations might
support more order constraints. We adress this problem following
the approach outlined in \citet{Beerenwinkel2006e} for the discrete CBN. 

Consider a family of posets $P_{\epsilon}$ ($0 \le \epsilon \le 1$),
each of which is maximal with the property that a fraction $\epsilon$ 
of the data is allowed to be incompatible with the poset. We assume that
the incompatible genotypes are generated with uniform probility
$q_{\epsilon} = 1/(2^n - |J(P_{\epsilon})|)$ and consider
the extended model with probabilities
\[
   \Pr(X_{\epsilon} \vdash Q \mid \alpha,\lambda) = 
   \begin{cases}
     \alpha \, \Pr(X_{\epsilon} \vdash Q \mid \lambda) & \text{if $Q$ refines $P_{\epsilon}$} \\
     (1 - \alpha) \, q_{\epsilon} & \text{else},
   \end{cases}
\]
where $X_{\epsilon} \sim P_{\epsilon}$ and 
$\alpha = \sum_{g \in J(P_{\epsilon})} u_g \big/ \sum_{g \in 2^{[n]}} u_g$ 
denotes the fraction of the data that are compatible with
$P_{\epsilon}$. The model can also be interpreted as a mixture model
with $\alpha$ the ML estimate of the mixing parameter
\citep[Prop.~8]{Beerenwinkel2006e}. In the applications, we construct
several posets $P_{\epsilon}$ for various values of $\epsilon$
and select the poset that maximizes the likelihood of the extended model.

\smallskip

In the first application, we analyze data from comparitive genome hybridization
(CGH) experiments. This technique detects large scale genomic alterations,
namely the gain or loss of chromosome arms, that occur frequently in cancer 
cells. For example, the event $4q+$ denotes the gain (+) of additional copies
of the large ($q$) arm of chromosome 4. Likewise, $8p-$ refers to the loss ($-$)
of the small arm ($p$) of chromosome 8. We consider 54 prostate cancer samples,
each defined by the presence or absence of the nine alterations
$3q+$, $4q+$, $6q+$, $7q+$, $8p-$, $8q+$, $10q-$, $13q+$, and $Xq+$ 
as defined in \citet{Rahnenfuehrer2005}. 

\begin{figure}[h!]
  \centering
  \includegraphics[width=0.7\textwidth]{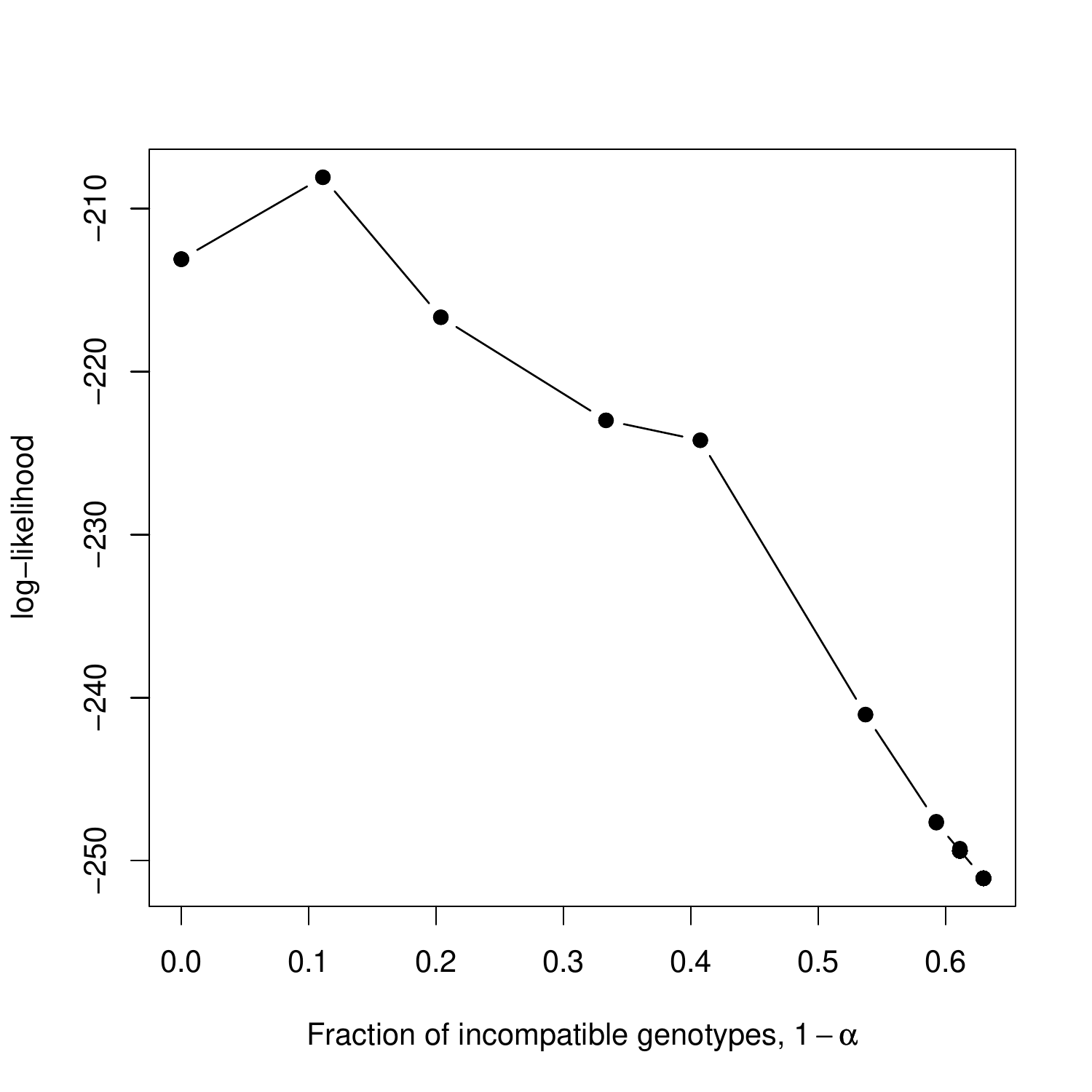}
  \caption{Maximum likelihood estimation of the discrete censored 
  CT-CBN model for the prostate cancer data.The log-likelihood is displayed 
  as a function of the fraction of data that is incompatible with the poset. 
  The curve has been generated by densely sampling $\epsilon$ from the
  unit interval and estimation of the extended models $P_\epsilon$.
  }
  \label{fig:pk}
\end{figure}

\begin{figure}[h!]
  \centering
  \makebox{
    \xymatrix@!=0.2cm{
      4q+\ar@{->}[d] &  \\
      8q+\ar@{->}[d]\ar@{->}[dr] & 6q+\ar@{->}[d]  \\
      13q+ & 3q+ & 7q+ & 8p- & 10q- & Xq+
    } 
  }
  \caption{Optimal prostate cancer poset corresponding to the maximum in Figure~\ref{fig:pk}.
  An arrow $p \rightarrow q$ between two genetic events represents the cover relation 
  $p \prec q$ in the Hasse diagram of the poset.
  }
  \label{fig:pk-poset}
\end{figure}
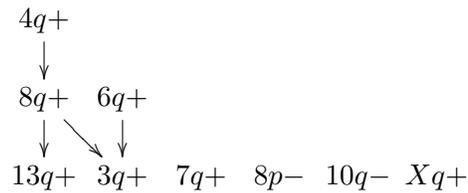

In Figure~\ref{fig:pk}, the log-likelihood is shown as a function
of the fraction of incompatible genotypes. The poset that maximizes
the likelihood explains 89\% of the data and is displayed in 
Figure~\ref{fig:pk-poset}. Four of the nine genetic changes do 
not obey any relation, two events have one predecessor, and one event 
occurs only after two parent events have occured. Note that the second
best poset is the empty poset corresponding to $\alpha = 0$.

\smallskip

In our second and third example, we consider mutation data from
35 colon and 42 breast cancer tumors, respectively
\citep{Sjoeblom2006}. Here, a genetic event is an unspecific
mutation in a gene that has been detected by DNA sequencing. 
Out of the $\sim\!200$ genes identified by \citet{Sjoeblom2006}
we considered those that were mutated in at least four tumors.
For colon cancer, this set comprises
{\em ADAMTSL3}, {\em APC}, {\em EPHA3}, {\em EPHB6}, {\em FBXW7}, 
{\em KRAS}, {\em MLL3}, {\em OBSCN}, {\em PKHD1}, {\em SMAD4}, 
{\em SYNE1}, and {\em TP53}, 
while for breast cancer, we identified
{\em ATP8B1}, {\em CUBN}, {\em FLJ13479}, {\em FLNB}, {\em MACF1}, 
{\em OBSCN}, {\em SPTAN1}, {\em TECTA}, and {\em TP53}.

\begin{figure}[h!]
  \centering
  \includegraphics[width=0.7\textwidth]{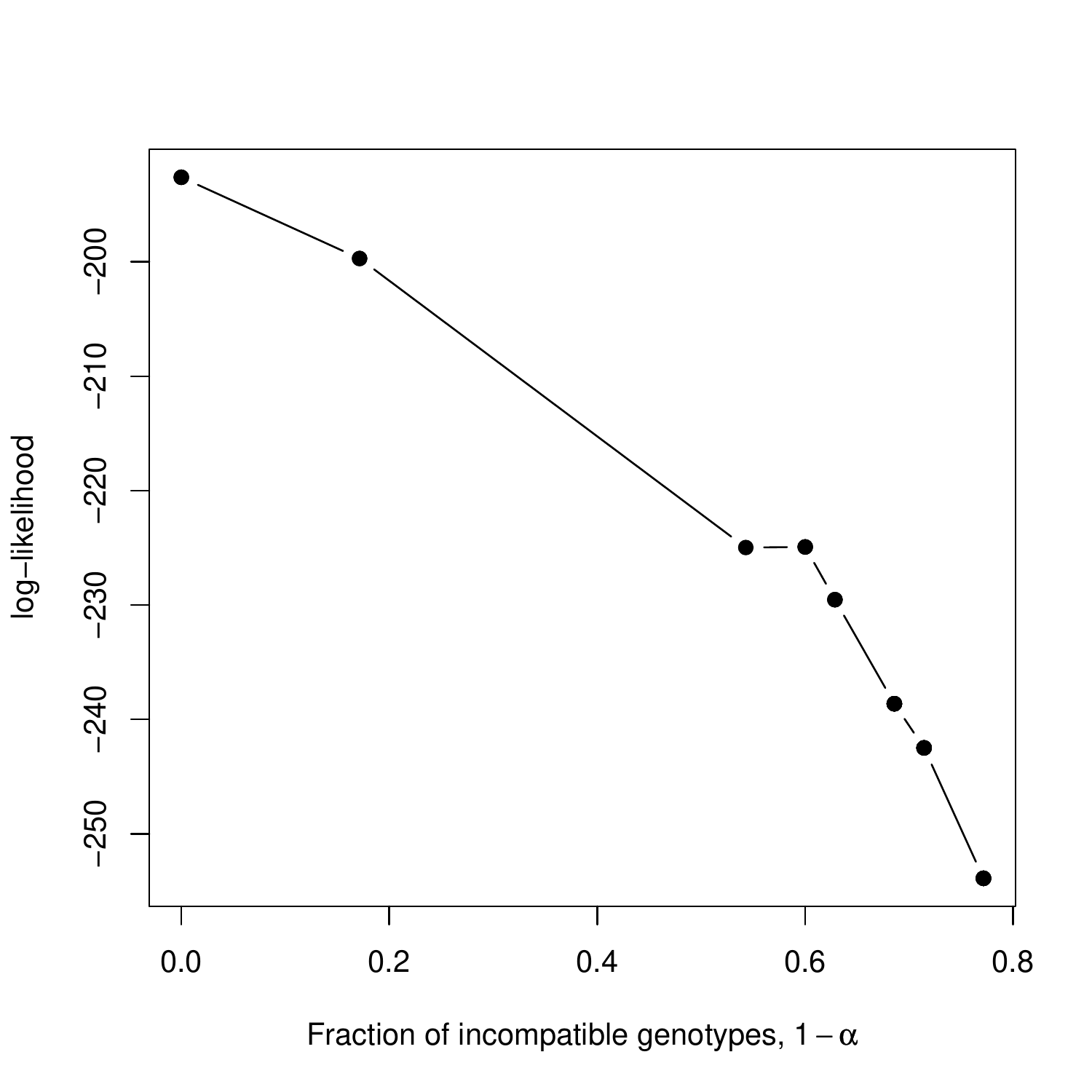}
  \caption{ML estimation of the CT-CBN model for the colon cancer data.
  }
  \label{fig:colon}
\end{figure}

\begin{figure}[h!]
  \centering
  \includegraphics[width=0.7\textwidth]{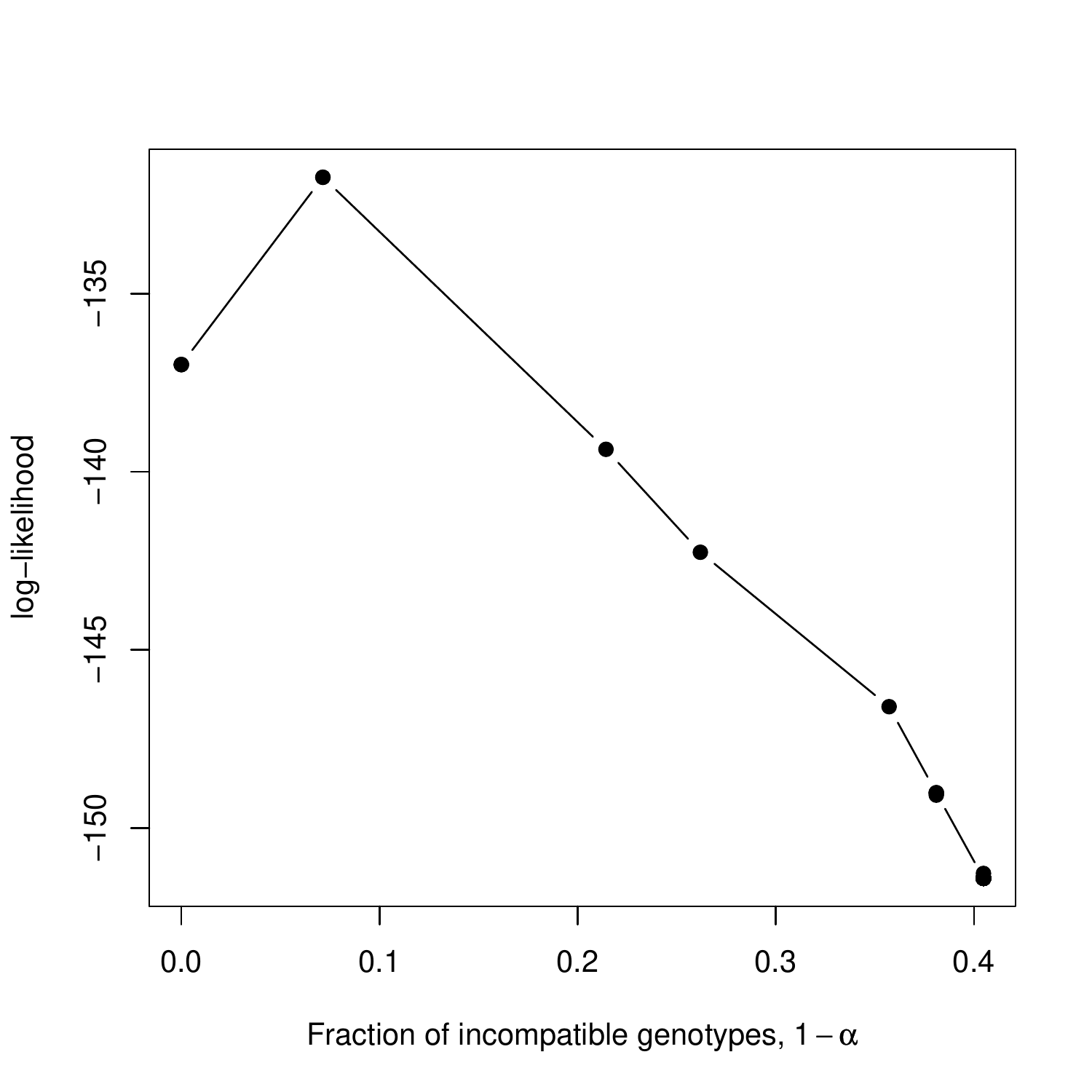} 
  \caption{ML estimation of the CT-CBN model for the breast cancer data.
  }
  \label{fig:breast}
\end{figure}

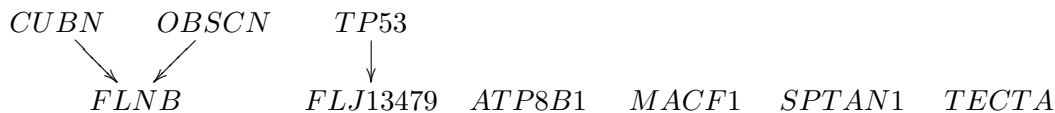
\begin{figure}[h!]
  \centering
  \makebox{
    \xymatrix@!=0.2cm{
      CUBN\ar@{->}[dr] &      & OBSCN\ar@{->}[dl] & & TP53\ar@{->}[d] \\
      & FLNB & & & FLJ13479 & & ATP8B1 & & MACF1 & & SPTAN1 & & TECTA
    }
    }
  \caption{Optimal breast cancer poset correponding to the maximum in .Figure~\ref{fig:breast}.
  }
  \label{fig:breast-poset}
\end{figure}

The colon cancer ML poset is the empty poset (Figure~\ref{fig:colon}).
By contrast, the  maximum likelihood breast cancer poset
explains 93\% of the observations (Figure~\ref{fig:breast})
and consists of three relations (Figure~\ref{fig:breast-poset}). 
Two of them form the conjuction stating that mutations in both 
the cubilin gene ({\em CUBN}) and the obscurin gene ({\em OBSCN}) 
occur before the $\beta$ filamin gene ({\em FLNB}) is
mutated. The third relation identifies tumor protein
p53 ({\em TP53}) as the mutational predecessor of the zinc finger
protein 668 ({\em FLJ13479}). {\em TP53} is mutated
in most colon and breast cancer tumors and known to
occur early in the somatic evolution of many cancers.

\smallskip

The last application is concerned with the evolution of drug
resistance in HIV. We study the accumulation of amino
acid changes in a segment of the HIV {\em pol} gene that
codes for the viral protein reverse transcriptase (RT).
The seven resistance-associated mutations 
41L, 67N, 69D, 70R, 210W, 215Y, and 219Q 
\citep{Johnson2006}
are considered, where, for example,  41L indicates the presence of
the amino acid leucine (L) at position 41 of the RT. 
A total of 364 viruses are analyzed that 
have been isolated from infected patients under therapy
with zidovudine, an antiretroviral drug targeting the RT.
Amino acid changes have been inferred after DNA sequencing of
the {\em pol} gene.

\begin{figure}[h!]
  \centering
  \includegraphics[width=0.7\textwidth]{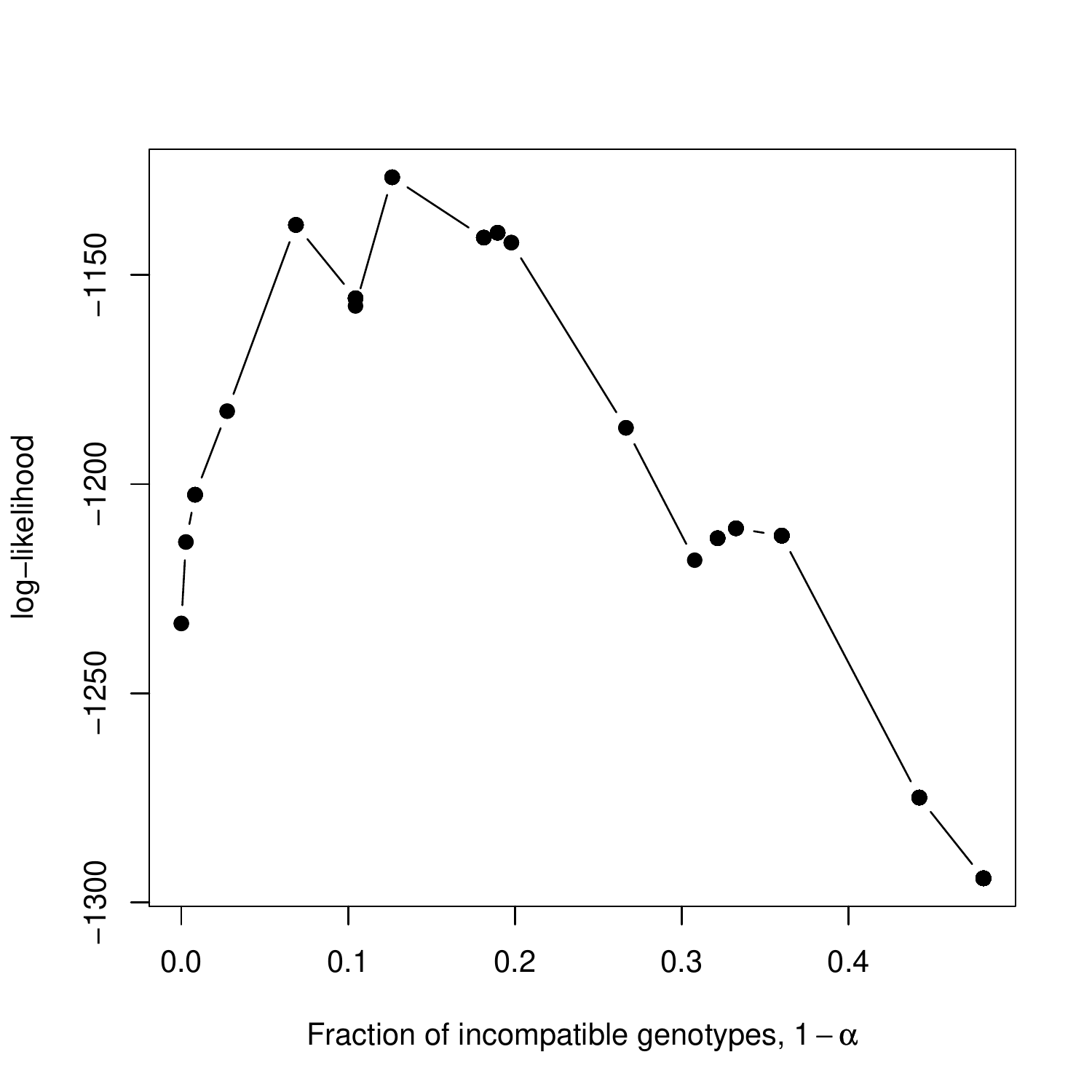}
  \caption{ML estimation of the CT-CBN model for the HIV drug resistance data.
  }
  \label{fig:zdv}
\end{figure}

\begin{figure}[h!]
  \centering
  \makebox{
    \xymatrix@!=0.2cm{
      215Y\ar@{->}[d] & & & 70R\ar@{->}[dl]\ar@{->}[dr] & \\
      41L\ar@{->}[d] & & 67N\ar@{->}[dr] & & 219Q\ar@{->}[dl] \\
      210W & & & 69D &
    }
  }
  \caption{Optimal HIV drug resistance poset correponding to the maximum in .Figure~\ref{fig:zdv}.
  }
  \label{fig:zdv-poset}
\end{figure}
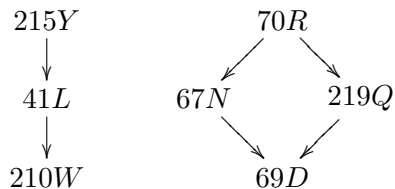

The optimal poset for the HIV drug resistance data explains
87\% of the observations (Figure~\ref{fig:zdv}). Its Hasse
diagram has two connected components (Figure~\ref{fig:zdv-poset}).
The first one represents the linear pathway 215Y $\prec$ 41L $\prec$ 210W,
whereas in the second component mutations 70R, 67N, 219Q,
and 69D, form a rhombus beginning with 70R and ending with 69D.
This finding confirms previous studies in which the same 
clustering of mutations has been described 
\citep{Beerenwinkel2005a,Boucher1992,Larder1989}.
The two groups of mutations are often referred to as the 
``215-41 pathway'' and the ``70-219 pathway'', respectively.
They provide alternative (but not exclusive) 
routes to resistance for HIV. The CT-CBN model captures
this escape behavior and suggests order constraints
within each group. 

\smallskip

In all applications discussed here, there are several posets
with near-optimal performance. Throughout we find these
posets to be very similar to the optimal set of relations.
For example, the second best HIV drug resistance poset
explains 93\% of the data and it differs from the optimal
one in Figure~\ref{fig:zdv-poset} only in the second
component which consists of the two relations 
70R $\prec$ 219Q and 70R $\prec$ 69D. Using cross-validation
techniques the variation in poset selection or in the
selection of individual relations can be studied and
may be used to select more robust poset structures.

For each of the four examples, the best mutagenetic trees
show inferior performance as compared to the CT-CBN posets. 
For example, the mutagenetic tree for prostate cancer found
in \citet[Fig.~2]{Rahnenfuehrer2005} explains only 56\% of 
the data at a log-likelihood of only $-248.4$. Unlike
mutagenetic trees, CT-CBNs can model the requirement of
multiple parent mutations. This type of order constraint 
was found in three of the four applications considered here. 


\section{Discussion}

Conjuctive Bayesian networks are statistical models for the accumulation 
of mutations. They are defined by a poset of mutations, which encodes
constraints on the order in which mutations can occur. Here we have
introduced the CT-CBN, a continuous time version of this model in which each
mutation appears after an exponentially distributed waiting time,
provided that all predecessor mutations in the poset have already
occurred. In an evolutionary process, this waiting time includes
the generation of the mutation plus the time it takes for the allele
to reach a frequency in the population that allows for its detection.
Since we consider only mutations with a selective advantage, the
waiting time will be dominated by the mutation process in large
populations and by the fixation process in small populations.
However, the model does not make any assumptions about the underlying
population dynamics. Hence the waiting time is the time to detection
of the mutation, which depends on the technology used for measurement.
For example, population sequencing of HIV samples can detect mutations
with a frequency of at least 20\%.

The CT-CBN informs about the order in which mutations tend to occur.
For a set of $n$ mutations, the number of possible pathways to evolve
the wild type into the type harboring all $n$ mutations is the number of
linear extensions of the poset. This cardinality increases rapidly
with $n$, but is hard to compute exactly \citep{Brightwell1991}. However,
evolution appears to follow only very few mutational paths to fitter
proteins \citep{Weinreich2006}. Thus, we generally expect to find posets
with much smaller lattices of order ideals than the full Boolean
lattice. Evolution takes place in this reduced genotype space
and can be modeled more efficiently. 

Estimation of the genetic event poset from observed data helps
understanding the phenotypic changes and biological mechanisms 
responsible for the fitness advantage. Furthermore, the poset allows
for identifying early and essential mutational steps that may
be predictive of clinical outcome or point to promising drug targets. 
In cancer research, \citet{Fearon1990} have proposed linear pathways 
of genetic alterations as a model of tumorigenesis. These models are 
known as ``Vogelgrams'' today \citep{Gatenby2003c}. The CT-CBN 
can be regarded as a generalization of the Vogelgram that is equipped 
with a statistical methodology for model selection and parameter estimation.
In particular, the CT-CBN allows for multiple evolutionary pathways and
makes explicit the timeline for the genetic alterations.

We have derived equations for the ML estimates of the model parameters and for
the expected waiting time of any genotype. These results are used in 
the EM algorithm for parameter estimation in the censored model. 
Censoring is modeled by assuming an exponentially distributed sampling 
time of the observed genotypes. This model appears most relevant for 
the data sets available, which often comprise cross-sectional data 
sampled after different but unknown time periods w.r.t.\ the 
evolutionary process. Other censoring schemes might be applicable in
the future and could also be worked out. For example, the sampling time
(but not the time of appearance of each mutation) can be observed in some 
situations, giving rise to a different marginalization of the fully
observed CT-CBN. 

Since model selection relies on a simple combinatorial criterion, and
the number of model parameters is only linear in the number of mutations,
we expect the CT-CBN to scale well with increasing data sets both in the number
of observations and the number of mutations.  In the cancer and HIV 
applications presented here, there are between 7 and 12 genetic events 
and 35 to 364 observations.  It is likely, however, that the number of 
genes associated with cancer progression, for example, is much higher 
than currently known \citep{Sjoeblom2006}. The running time of the EM
algorithm is dominated by the size of the distributive lattice of order
ideals. Thus, many mutations can be modeled as long as the number of mutational
pathways is limited.

\section*{Acknowledgments}  

Seth Sullivant was supported by NSF grant DMS-0700078.
Part of this work was done while Niko Beerenwinkel was
affiliated with the Program for Evolutionary Dynamics
at Harvard University and funded by a grant from the 
Bill \& Melinda Gates Foundation through the Grand Challenges 
in Global Health Initiative.

\bibliographystyle{abbrvnat}
\bibliography{nikos}

%
%
%
%
%
%
%
%

\end{document}